\documentclass[a4paper]{article}
\usepackage[utf8]{inputenc}
\usepackage[OT4]{fontenc}
\usepackage{amsthm}
\usepackage{amsmath}
\usepackage{amssymb}
\usepackage{enumerate}
\usepackage{graphicx}
\usepackage{fullpage}
\usepackage{subfig}

\newcommand{\maybeqed}{}

\newcommand{\dual}[1]{#1^{*}}
\newcommand{\dc}[1]{D_#1}
\newcommand{\boundary}[1]{\partial(#1)}
\newcommand{\rdiv}{\mathcal{P}}
\newtheorem{definition}{Definition}[section]
\newtheorem{lemma}[definition]{Lemma}

\newtheorem{theorem}[definition]{Theorem}

\title{Optimal decremental connectivity in planar graphs}
\author{Jakub Łącki \quad Piotr Sankowski \bigskip \\
    \texttt{\{j.lacki,sank\}@mimuw.edu.pl}\\
    University of Warsaw,\\
    Warsaw, Poland}

\begin{document}

\maketitle

\begin{abstract}
We show an algorithm for dynamic maintenance of connectivity information in an undirected planar graph subject to edge deletions.
Our algorithm may answer connectivity queries of the form `Are vertices $u$ and $v$ connected with a path?' in constant time.
The queries can be intermixed with any sequence of edge deletions, and the algorithm handles all updates in $O(n)$ time.
This results improves over previously known $O(n \log n)$ time algorithm.
\end{abstract}

\section{Introduction}

The \emph{dynamic graph connectivity} problem consists in maintaining connectivity information about an undirected graph, which is undergoing modifications.
Typically, the modifications are additions or removals of edges or vertices.
In this paper we focus on the problems in which each modification adds or removes a single edge.
These problems have three variants: in the \emph{incremental} version, edges can only be added to the graph, in the \emph{decremental} one the edges may only be removed, whereas in the \emph{fully dynamic} version both edge insertions and deletions are allowed.
Graph updates are intermixed with a set of connectivity queries of the form `Are vertices $u$ and $w$ in the same connected component?'

We consider the decremental connectivity problem for planar graphs, and show an algorithm that may answer connectivity queries in constant time and process any sequence of edge deletions in $O(n)$ time. The previously known best running time of $O(n \log n)$ was obtained by using the fully dynamic algorithm.

\subsection{Prior work}
It is easy to see that the incremental graph connectivity can be solved using an algorithm for the union-find problem.
It follows from the result of Tarjan~\cite{Tarjan75} that a sequence of $n$ edge insertions and $n$ queries can be handled in $O(n \alpha(n))$ time, where $\alpha(n)$ is the extremely slowly growing inverse Ackermann function.

There has been a long line of research considering the fully dynamic connectivity in general graphs~\cite{Frederickson85, Eppstein, Henzinger99, Holm01, Thorup00, Kapron13, Nilsen13}.
The best currently known algorithms have polylogartithmic update and query time.
Thorup~\cite{Thorup00} has shown a randomized algorithm with $O(\log n (\log \log n)^3)$ amortized update and $O(\log n / \log \log \log n)$ query time.
An algorithm by Wulff-Nilsen~\cite{Nilsen13} handles updates in slightly worse $O(\log^2 n / \log \log n)$ amortized time, but it is deterministic and answers queries in $O(\log n / \log \log n)$ time.
The best algorithm with worst-case update guarantee is a randomized algorithm by Kapron, King and Mountjoy~\cite{Kapron13}, which processes updates in $O(\log^5 n)$ time and answers queries in $O(\log n / \log \log n)$ time.

For the decremental variant, Thorup~\cite{Thorup99} has shown a randomized algorithm, which can process any sequence of edge deletions in $O(m \log(n^2 / m) + n (\log n)^3 (\log \log n)^2)$ time and answers queries in constant time.
If $m = \Theta(n^2)$, the update time is $O(m)$, whereas for $m = \Omega(n (\log n \log \log n)^2)$ it is $O(m \log n)$.

The picture is much simpler in case of planar graphs.
Eppstein et.\ al~\cite{Eppstein92} gave a fully dynamic algorithm which handles updates and queries in $O(\log n)$ amortized time, but requires that the graph embedding remains fixed.
For the general case (i.e., when the embedding may change) Eppstein et.\ al~\cite{Eppstein96} gave an algorithm with  $O(\log^2 n)$ amortized update time and $O(\log n)$ query time.

In planar graphs, the best known solution for the incremental connectivity problem is the union-find algorithm.
However, for the special case when the final resulting planar graph is given upfront, and after that the edge insertions and queries are given in a dynamic fashion Gustedt~\cite{Gustedt} has shown an $O(n)$ time algorithm.
On the other hand, for the decremental problem nothing better than a direct application of the fully dynamic algorithm is known.
This is different from both general graphs and trees, where the decremental connectivity problems have better solutions than what could be achieved by a simple application of their fully dynamic counterparts.
In case of general graphs, the best total update time was $O(m \log n)$~\cite{Thorup99} (except for very sparse graphs, including planar graphs), compared to $O(m \log n (\log \log n)^3)$ time for the fully dynamic variant.
For trees, only $O(n)$ time is necessary to perform all updates in the decremental scenario~\cite{Alstrup97}, while in the fully dynamic case one can use dynamic trees and obtain $O(\log n)$ worst case update time. 

There has also been some progress in obtaining lower bounds for dynamic connectivity problems.
Tarjan and La Poutr{\'{e}}~\cite{Tarjan79, Poutre96} have shown that incremental connectivity requires $\Omega(\alpha(n))$ time per operation on a pointer machine.
Henzinger and Fredman~\cite{Henzinger98} considered the fully dynamic problem and RAM model and obtained a lower bound of $\Omega(\log n / \log \log n)$, which also works for plane graphs.
This was improved by Demaine and P{\v a}tra{\c s}cu~\cite{Patrascu06} to a lower bound of $\Omega(\log n)$ in cell-probe model.
The lower bound holds also for plane graphs.

\subsection{Our results}
We show an algorithm for the decremental connectivity problem in planar graphs, which processes any sequence of edge deletions in $O(n)$ time and answers queries in constant time.
This improves over the previous bound of $O(n \log n)$, which can be obtained by applying the fully dynamic algorithm by Eppstein~\cite{Eppstein92}, and matches the running time of decremental connectivity on trees~\cite{Alstrup97}.

In fact, we present a $O(n)$ time reduction from the decremental connectivity problem to a collection of incremental problems in graphs of total size $O(n)$.
These incremental problems have a specific structure: the set of allowed union operations forms a planar graph and is given in advance.
As shown by Gustedt~\cite{Gustedt}, such a problem can be solved in linear time. 
Our result shows that in terms of total update time, the decremental connectivity problem in planar graphs is definitely not harder than incremental one.
It should be noted that union-find algorithm can process any sequence of $k$ query or update operations in $O(k \alpha(n))$ time, while our algorithm requires $O(n)$ time to process any sequence of edge deletions and answers queries in constant time.

Moreover, since the fully dynamic connectivity has a lower bound of $\Omega(\log n)$ (even in plane graphs) shown by Demaine and P{\v a}tra{\c s}cu~\cite{Patrascu06}, our results implies that in planar graphs decremental connectivity is strictly easier than fully dynamic one.
We suspect that the same holds for general graphs, and we conjecture that it is possible to break the $\Omega(\log n)$ bound for a single operation of a decremental connectivity algorithm, or the $\Omega(m \log n)$ bound for processing a sequence of $m$ edge deletions. 

Our algorithm, unlike the majority of algorithms for maintaining connectivity, does not maintain the spanning tree of the current graph.
As a result, it does not have to search for a replacement edge when an edge from the spanning tree is deleted.
It is based on a novel and very simple approach for detecting bridges, which alone gives $O(n \log n)$ total time.
We use the fact that a deletion of edge $uw$ in the graph causes some connected component to split if both sides of $uw$ belong to the same face.
This condition can in turn be verified by solving an incremental connectivity problem in the dual graph.
When we detect a deletion that splits a connected component, we start two parallel DFS searches from $u$ and $w$ to identify the \emph{smaller} of the two new components. Once the first search finishes, the other one is stopped.
A simple argument shows that this algorithm runs in $O(n \log n)$ time.

We then show that the DFS searches can be speeded up using an $r$-division, that is a decomposition of a planar graph into subgraphs of size at most $r = \log^2 n$. This gives an algorithm running in $O(n \log \log n)$ time. For further illustration of this idea we show how to apply it recursively
in order to obtain an $O(n \log^{*} n)$ time algorithm. However, we observe that it is enough to use this recursion only twice.
This is because the $O(n \log \log n)$ time algorithm, as an intermediate step, reduces the problem of maintaining connectivity in the input graph to maintaining connectivity in a number of graphs of size at most $r = \log^2 n$. By using this reduction twice, we reduce the problem to graphs of size $O(\log^2 \log n)$. The number of such graphs is so small that we can simply precompute the answers for all of them and use these precomputed answers to obtain the main result of the paper. The preprocessing of all graphs of bounded size is again an idea that, to the best of our knowledge, has never been previously used for designing dynamic graph algorithms.

\subsection{Organization of the paper}
In Section~\ref{sec:preliminaries} we introduce notation and recall some of the concepts that we later use.
The following sections describe our algorithm.
We start with the description of the simple $O(n \log n)$ time algorithm in Section~\ref{sec:n_log_n}, and then in every section we show an improvement in the running time.

In Section~\ref{sec:n_log_log_n} we show how to use $r$-division to get an $O(n \log \log n)$ algorithm.
Section~\ref{sec:n_log_star_n}, shows how to improve the reduction, so that it can be used more than once, which results in an $O(n \log^{*} n)$ time algorithm. Finally, in Section~\ref{sec:n} we show how to solve the decremental connectivity in optimal time for graphs of size $O(\log^2 \log n)$, after initial preprocessing. This, combined with the reduction applied twice, gives the main result of the paper.

\section{Preliminaries}
\label{sec:preliminaries}
Let $G=(V, E)$ be an undirected, unweighted planar graph, and $n = |V|$.
By $V(G)$, $E(G)$ and $F(G)$ we denote the sets of vertices, edges and faces of $G$.
By Euler's formula $|V(G)| - |E(G)| + |F(G)| = |CC(G)| + 1$, where $CC(G)$ is the set of connected components of $G$.
The \emph{dual graph} $\dual{G}$ is constructed from $G$ by embedding a single vertex in every face of $G$ and connecting the vertices in adjacent faces of $G$.
Note that if two faces $f_1$, $f_2$ share more than one edge, $\dual{G}$ has multiple edges between $f_1$ and $f_2$.

In the paper we deal with algorithms that maintain the connectivity information about a graph $G$ subject to edge deletions.
By the total running time we denote the total time of handling deletions of all edges from the graph.

The identifier of a connected component (henceforth denoted \emph{cc-identifier}) is a value assigned to a vertex $v \in V$, which uniquely identifies the connected component of $G$, i.e., two vertices have the same cc-identifier if and only if they belong to the same connected component.
The cc-identifiers change as the edges are deleted, and they may not be preserved after edge deletion.
An algorithm maintains cc-identifiers \emph{explicitly} if after every deletion it returns the list of changes to the cc-identifiers.
We assume that cc-identifiers are $O(\log n)$-bit integers.
Note that an algorithm which maintains cc-identifiers explicitly can be simply turned into an algorithm with constant query time.
In order to answer a query regarding two vertices, it suffices to compare the cc-identifiers of the two vertices.
By definition, the vertices are in the same connected component if and only if their cc-identifiers are equal.

Let us now recall the notion of an $r$-division.
A \emph{region} $R$ is an edge-induced subgraph of $G$.
A \emph{boundary vertex} of a region $R$ is a vertex $v \in V(R)$ that is adjacent to an edge $e \not\in E(R)$.
We denote the set of boundary vertices of a region $R$ by $\boundary{R}$.
An \emph{$r$-division} $\rdiv$ of $G$ is a partition of $G$ into $O(n/r)$ edge-disjoint regions (which might share vertices), such that each region contains at most $r$ vertices and $O(\sqrt{r})$ boundary vertices.
The set of boundary vertices of a division $\rdiv$, denoted $\boundary{\rdiv}$ is the union of the sets $\boundary{R}$ over all regions $R$ of $\rdiv$.
Note that $|\boundary{\rdiv}| = O(n / \sqrt{r})$.

\begin{lemma}[\cite{Klein-13,Arge-13}]\label{lem:r-division}
Let $G=(V,E)$ be an $n$-vertex biconnected triangulated planar graph and $1 \leq r \leq n$.
An $r$-division of $G$ can be constructed in $O(n)$ time.
\end{lemma}

Let $G$ be a planar graph.
In the preprocessing phase of our algorithms, we build an $r$-division of $G$.
This $r$-division will be updated in a natural way, as edges are deleted from $G$.
Namely, when an edge is deleted from the graph, we update its $r$-division by deleting the corresponding edge.
However, if we strictly follow the definition, what we obtain may no longer be an $r$-division.

For that reason, we loosen the definition of an $r$-division, so that it includes the divisions obtained by deleting edges.
Consider an $r$-division $\rdiv$ built for a graph $G$.
Moreover, let $G'$ be a graph obtained from $G$ by deleting edges, and let $\rdiv'$ be the $r$-division $\rdiv$ updated in the following way.
Let $R$ be a region of $\rdiv$.
Then, we define the graph $R'$ in $\rdiv$ obtained by removing edges from $R$ to be a region of $\rdiv'$, although it may no longer be an edge-induced subgraph of $G'$, e.g., it may contain isolated vertices.
Similarly, we define the set of boundary vertices of $\rdiv'$ to be the set of boundary vertices of $\rdiv$.
Again, according to this definition, a boundary vertex $v$ of $\rdiv'$ may be incident to edges of a single region (because the edges incident to $v$ that belonged to other regions have been deleted).
In the following, we say that $\rdiv'$ is an $r$-division of $G'$.

Since Lemma~\ref{lem:r-division} requires the graph to be biconnected and triangulated, in order to obtain an $r$-division for a graph which does not have these properties, we first add edges to $G$ to make it biconnected and triangulated, then compute the $r$-division of $G$, and then delete the added edges both from $G$ and its division.

Without loss of generality, we can assume that each vertex $v \in V$ has degree at most $3$.
This can be assured by triangulating the dual graph in the very beginning.
In particular, this assures that each vertex belongs to a constant number of regions in an $r$-division.

We assume that all logarithms we use are binary.
We define $\log^{(0)} n := n$ and, for $t > 1$ $\log^{(t)} n := \log^{(t-1)} \log n$.
Moreover, we define the iterated logarithm
 $\log^{*} n := \min \{t : t \in \mathbb{N}, \log^{(t)} n \leq 1\}$.

\section{$O(n \log n)$ time algorithm}
\label{sec:n_log_n}
Let $G$ be a planar graph subject to edge deletions.
We call an edge deletion \emph{critical} if and only if it increases the number of components of $G$, i.e., the deleted edge is a bridge in $G$.
We first show a dynamic algorithm that for every edge deletion decides, whether it is critical.
It is based on a simple relation between the graph $G$ and its dual.

\begin{lemma}
\label{lem:critical_deletions}
Let $G$ be a planar graph subject to edge deletions.
There exists an algorithm that for each edge deletion decides whether it is critical.
It runs in $O(n)$ total time.
\end{lemma}

\begin{figure}
\begin{center}
\includegraphics[scale=0.5]{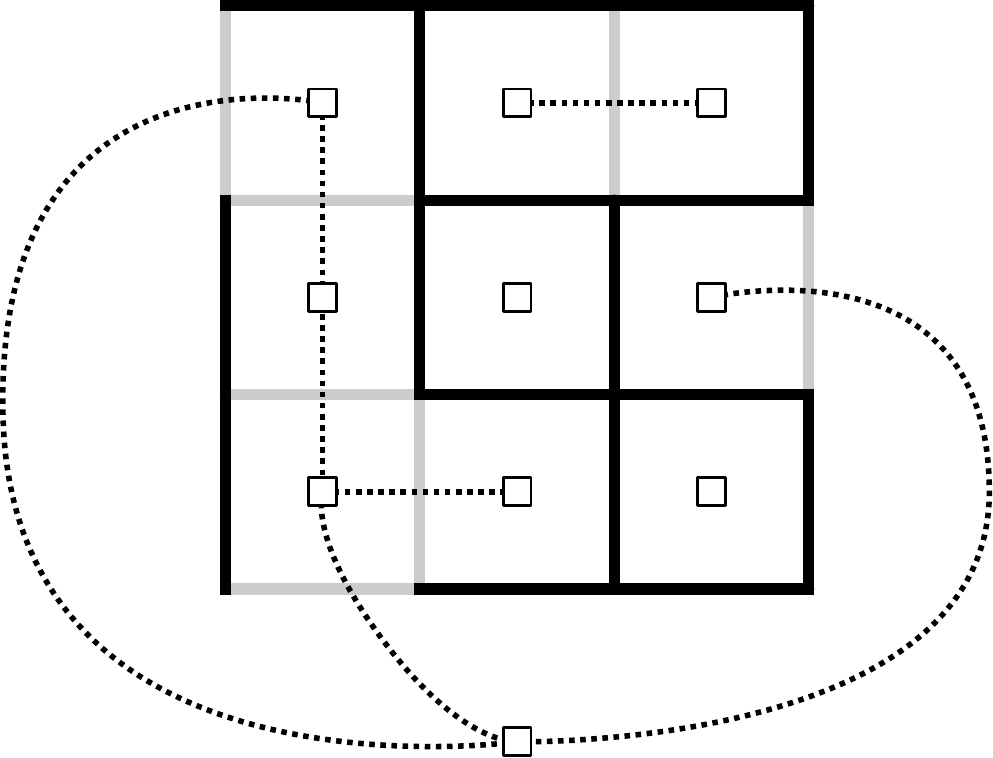}
\end{center}
\caption{\label{fig:dual_complement}The graphs from the proof of Lemma~\ref{lem:n_log_n}.
Edges of $G$ are drawn with solid black lines, whereas the gray lines depict edges that have been deleted from $G$.
The small squares are vertices of $\dc{G}$, and the dotted lines are edges of $\dc{G}$.}
\end{figure}

\begin{proof}
We will maintain the number of faces in $G$.
When an edge $e$ is deleted, we simply have to merge faces on both sides of $e$ (if they are different from each other).
This can be implemented using union-find data structure on the vertices of the dual graph.

More formally, we build and maintain a graph $\dc{G}$.
Initially, this is a graph consisting of vertices of $\dual{G}$ (faces of $G$).
When an edge is deleted from $G$, we add its dual edge to $\dc{G}$ (see Fig.~\ref{fig:dual_complement}).
Clearly, the connected components of $\dc{G}$ are exactly the faces of $G$.
Since edges are only added to $\dc{G}$, we can easily maintain the number of connected components in $\dc{G}$ with a union-find data structure.

This allows us to detect critical deletions in $G$.
After every edge deletion, we know the number of edges and vertices of $G$.
Moreover, we know that the number of faces of $G$ is equal to the number of connected components of $\dc{G}$, which we also maintain.
As a result, by Euler's formula, we get the number of connected components of $G$, so in particular we may check if the deletion caused the number of connected components to increase.
The algorithm executes $O(n)$ find and union operations on the union-find data structure.

However, the sequence of union operations has a certain structure.
Let $G_1$ be the initial version of the graph $G$ (before any edge deletion).
Observe that each union operation takes as arguments the endpoints of an edge of $\dual{G_1}$.
The variant of the union-find problem, in which the set of allowed union operations forms a planar graph given during initialization, was considered by Gustedt~\cite{Gustedt}.
He showed that for this special case of the union-find problem there exists an algorithm that may execute any sequence of $O(n)$ operations in $O(n)$ time, given an $n$-vertex planar graph.
Thus, we infer that our algorithm runs in $O(n)$ time.
\end{proof}

We can now use Lemma~\ref{lem:critical_deletions} to show a simple decremental connectivity algorithm that runs in $O(n \log n)$ total time.

\begin{lemma}
\label{lem:n_log_n}
Let $G$ be a planar graph subject to edge deletions.
There exists a decremental connectivity algorithm that for every vertex of $G$ maintains its cc-identifier explicitly.
It runs in $O(n \log n)$ total time.
\end{lemma}

\begin{proof}
We use Lemma~\ref{lem:critical_deletions} to detect critical deletions.
When an edge $uw$ is deleted, and the deletion is not critical, nothing has to be done.
Otherwise, after a critical deletion, some connected component $C$ breaks into two components $C_u$ and $C_w$ ($u \in C_u$, $w \in C_w$) and we start two parallel depth-first searches from $u$ and $w$.
We stop both searches once the first of them finishes.
W.l.o.g. assume that it is the search started from $u$.
Thus, we know that the size of $C_u$ is at most half of the size of $C$.\footnote{Since the graph has constant degree, we may assure that both searches are synchronized in terms of number of vertices visited.}
We can now iterate through all vertices of $C_u$ and change their cc-identifiers to a new unique number.
All these steps require $O(|C_u|)$ time.
The running time of the algorithm is proportional to the total number of changes of the cc-identifiers.
Since every vertex changes its identifier only when the size of its connected component halves, we infer that the total running time is $O(n \log n)$.
\end{proof}

\section{$O(n \log \log n)$ time algorithm}
\label{sec:n_log_log_n}
In order to speed up the $O(n \log n)$ algorithm, we need to speed up the linear depth-first searches that are run after a critical edge deletion.
We build an $r$-division $\rdiv$ of $G$ for $r=\log^2 n$ and use a separate decremental connectivity algorithm to maintain the connectivity information inside each region.
On top of that, we maintain a \emph{skeleton graph} that represents connectivity information between the set of boundary vertices (and possibly some other vertices that we consider important).
Loosely speaking, since the number of boundary vertices is $O(n / \log n)$ we can pay a cost of $O(\log n)$ for maintaining the cc-identifier for each of them.

\begin{figure}
\begin{minipage}{.4\linewidth}
\centering
\subfloat[\label{sub:graph_g}]{\includegraphics[scale=0.6]{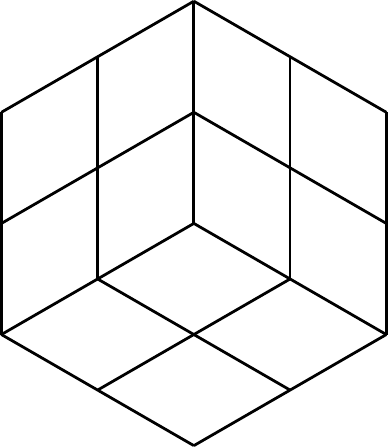}}
\end{minipage}
\begin{minipage}{.4\linewidth}
\centering
\subfloat[\label{sub:g-rdiv}]{\includegraphics[scale=0.6]{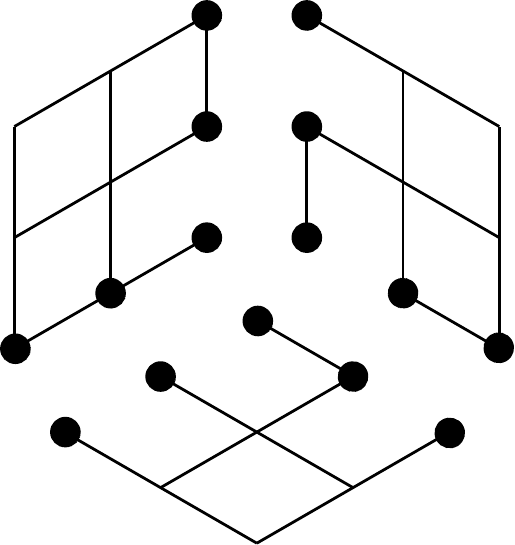}}
\end{minipage}
\medskip
\centering

\hspace{0.1\textwidth}
\subfloat[\label{sub:g_del}]{\includegraphics[scale=0.6]{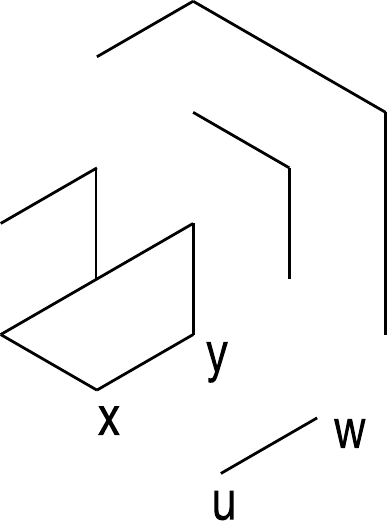}}
\hfill
\subfloat[\label{sub:g_del_rdiv}]{\includegraphics[scale=0.6]{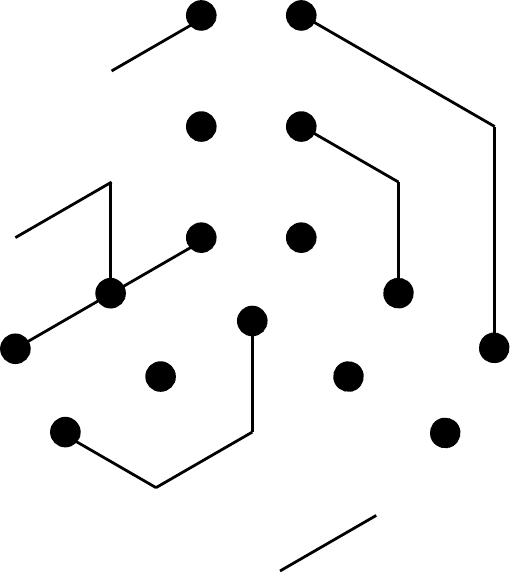}}
\hfill
\subfloat[\label{sub:skeleton}]{\includegraphics[scale=0.6]{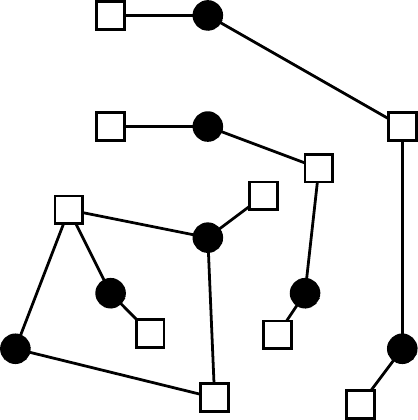}}
\hspace{0.1\textwidth}
\caption{\label{fig:skeleton_graph}Panels~\ref{sub:graph_g} and~\ref{sub:g-rdiv} show a sample graph $G$ and its $r$-division into three regions (boundary vertices are marked with small circles).
    In panel~\ref{sub:g_del} there is graph $G'$ obtained from $G$ by a sequence of edge deletions.
    Panel~\ref{sub:g_del_rdiv} shows its $r$-division obtained from the $r$-division of $G$.
    Finally, panel~\ref{sub:skeleton} contains the skeleton graph of $G'$. 
    Auxiliary vertices are marked with squares.}
\end{figure}

\begin{definition}
Consider an $r$-division $\rdiv$ of a planar graph $G = (V,E)$ and a set $V_s$ (called a \emph{skeleton set}), such that $\boundary{\rdiv} \subseteq V_s \subseteq V$.
The \emph{skeleton graph} for $\rdiv$ and $V_s$ is a graph over the skeleton set $V_s$ and some additional auxiliary vertices.
Consider a region $R$ of $\rdiv$.
Group vertices of $V_s \cap V(R)$ into sets $V_1, \ldots, V_k$, such that two vertices belong to the same set if and only if there is a path in $R$ that connects them.
For each set $V_i$ add a new auxiliary vertex $w_i$ and add an edge $w_ix$ for every $x \in V_i$.
\end{definition}

For illustration, see Fig.~\ref{fig:skeleton_graph}.
Observe that the skeleton graph has $O(|V_s|)$ vertices and edges.
Moreover, if $u, w \in V_s$, then $u$ and $v$ are connected in the skeleton graph if and only if they are connected in $G$.
The skeleton graph is also planar, but our algorithms do not use this property.

In our algorithm we will update the skeleton graph of $G$, as edges are deleted.
Similarly to the $O(n \log n)$ algorithm, we need a way of detecting whether an edge deletion in $G$ increases the number of connected components in the skeleton graph.

\begin{lemma}
\label{lem:skeleton_critical}
Let $G$ be a dynamic planar graph, subject to edge deletions.
Assume that we maintain its skeleton graph $G_s$ computed for an $r$-division $\rdiv$ and a skeleton set $V_s$.
An edge deletion in $G$ causes an increase in the number of connected components in $G_s$ if and only if the deletion is critical in $G$ and there exists a region of $\rdiv$, in which the deletion disconnects some two vertices of $V_s$.
\end{lemma}

Before we proceed with the proof, let us note that all its conditions are necessary.
In particular, a critical deletion in $G$ may not disconnect some two vertices of a skeleton set in a region (e.g. edge $uw$ in Fig.~\ref{sub:g_del}, whose deletion does not affect the skeleton graph at all).
It may also happen that the deletion is not critical in $G$, but inside some region it disconnects some two vertices of $V_s$ (e.g. edge $xy$ in Fig.~\ref{sub:g_del}).

\begin{proof}
Recall that two vertices of $V_s$ are connected in $G$ iff they are connected in $G_s$.

\noindent ($\implies$) If two vertices of $V_s$ become disconnected in $G_s$, they also become disconnected in $G$, so the edge deletion is critical.
The deletion has to disconnect some two vertices in a region, because otherwise the graph $G_s$ would not change at all.

\noindent ($\impliedby$)
Assume that the deletion disconnected vertices $u, w \in V_s$ in a region $R$.
Thus, the deleted edge was on some path from $u$ to $w$.
Since the edge deletion is critical in $G$, the deleted edge was a bridge in $G$.
After the deletion there is no path from $u$ to $w$ in $G$ and consequently also in $G_s$.
\maybeqed
\end{proof}

We are ready to show the main building block of our $O(n \log \log n)$ algorithm.

\begin{lemma}
\label{lem:first_recursion}
Let $G$ be a planar graph.
Assume there exists a decremental connectivity algorithm that runs in $f(n)$ time and maintains cc-identifiers explicitly.
Then, there exists a decremental connectivity algorithm that runs in $O(n +  n \cdot f(\log^2 n) / \log^2 n)$ time and answers queries in $O(1)$ time.
\end{lemma}

\begin{proof}
We build an $r$-division $\rdiv$ of $G$ for $r = \log^2 n$.
By Lemma~\ref{lem:r-division}, this takes $O(n)$ time.
For each region of the division, we run the assumed decremental algorithm to handle edge deletions.
Moreover, we use Lemma~\ref{lem:critical_deletions} to detect critical deletions in $G$.

We build the skeleton graph $G_s$ for $G$, $r$-division $\rdiv$ and a skeleton set $V_s = \boundary{\rdiv}$.
We maintain $G_s$, as edges are deleted, that is the deletions to $G$ are reflected in $G_s$. 
This can be done using the decremental algorithms that we run for every region.
Since they maintain the cc-identifiers explicitly (we call these identifiers \emph{local} cc-identifiers), we may detect the moment when some two vertices of $V_s$ become disconnected within one region and $G_s$ needs to be updated.
Note that if a deletion causes $t$ cc-identifiers to change, we may update $G_s$ in $O(t)$ time, so the time for updating $G_s$ is linear in the number of local cc-identifiers that are changed.

For every vertex of $G_s$, we maintain its cc-identifier (called a \emph{global} cc-identifier).
Once $G_s$ is updated after an edge deletion, we use Lemma~\ref{lem:skeleton_critical} to check whether the number of connected components of $G_s$ increased.
According to the lemma, it suffices to check whether the deletion is critical in $G$ (this is reported by the algorithm of Lemma~\ref{lem:critical_deletions}), and whether some two skeleton vertices became disconnected within some region (this can be checked easily by inspecting the changes of the cc-identifiers).

When we detect that the number of connected components of the skeleton graph $G_s$ has increased, similarly to the $O(n \log n)$ algorithm we run two parallel DFS searches to identify the smaller of the two new connected components.
After that, we update the global cc-identifiers.

In order to answer a query regarding two vertices $u$ and $w$, we perform two checks.
First, if the vertices belong to the same region, we check whether there exists a path connecting them that does not contain any boundary vertices.
This can be done by querying the decremental algorithm for the appropriate region.

Then, we check whether there is a path from $u$ to $w$ that that contains some boundary vertex.
For each of the two vertices, we find two arbitrary boundary vertices $b_u$ and $b_w$ that $u$ and $w$ are connected to (note that with no additional overhead we may maintain, for each region and each local cc-identifier, a list of boundary vertices with this cc-identifier).
Then, we check whether $b_u$ and $b_w$ have the same global cc-identifier.

Let us now analyze the running time.
The algorithm of Lemma~\ref{lem:critical_deletions} requires $O(n)$ time.
The decremental algorithms run inside regions take $O(n \cdot f(r) / r) = O(n \cdot f(\log^2 n) / \log^2 n)$ time.
Lastly, we bound the running time of the DFS searches performed to update the global cc-identifiers.
We use an argument similar to the one in the proof of Lemma~\ref{lem:n_log_n}.
The skeleton graph has $O(n / \log n)$ vertices, and each global cc-identifier can change at most $O(\log (n / \log n)) = O(\log n)$ times.
Hence, the DFS searches require $O((n / \log n) \log n) = O(n)$ time.
The lemma follows.
\maybeqed
\end{proof}

By applying Lemma~\ref{lem:n_log_n} to Lemma~\ref{lem:first_recursion}, we obtain the following.

\begin{lemma}
There exists a decremental connectivity algorithm for planar graphs that runs in $O(n \log \log n)$ total time.
\end{lemma}

\begin{proof}
Since $f(n) = O(n \log n)$, the running time is $O(n + n \cdot f(\log^2 n) / \log^2 n) = O(n + n \log^2 n \log \log n / \log^2 n) = O(n \log \log n)$.
\end{proof}

\section{$O(n \log^{*} n)$ time algorithm}
\label{sec:n_log_star_n}
In order to obtain a faster algorithm, we would like to use Lemma~\ref{lem:first_recursion} multiple times, starting from the $O(n \log n)$ algorithm, and each time applying the lemma to the algorithm obtained in the previous step.
This, however, cannot be done directly.
While the lemma requires an algorithm that maintains all cc-identifiers explicitly, it does not produce an algorithm with this property.
We deal with this problem in this section.

Observe that in the proof of Lemma~\ref{lem:first_recursion} we only used the algorithms to maintain the cc-identifiers of the vertices of the skeleton set. We show that we can adapt our algorithms to maintain only some of the cc-identifiers.

\begin{lemma}
\label{lem:second_recursion}
Assume there exists a decremental connectivity algorithm for planar graphs that, given a graph $G=(V,E)$ and a set $V_e \subseteq V$ (called an \emph{explicit} set):
\begin{itemize}
\item maintains cc-identifiers of the vertices of $V_e$ explicitly,
\item processes updates in $f(n) + O(|V_e| \log n)$ time,
\item may return the cc-identifier of any vertex in $g(n)$ time,
\end{itemize}
where $f(n)$ and $g(n)$ are nondecreasing functions.

Then, there exists a decremental connectivity algorithm for planar graphs, which, given a graph $G=(V,E)$ and a set $V_e \subseteq V$:
\begin{itemize}
\item maintains cc-identifiers of the vertices of $V_e$ explicitly,
\item processes updates in $O(n + |V_e| \log n + n \cdot f(\log^2 n) / \log^2 n) $ time,
\item may return the cc-identifier of any vertex in $g(\log^2 n) + O(1)$ time.
\end{itemize}
\end{lemma}

\begin{proof}
We build an $r$-division $\rdiv$ of $G$ for $r = \log^2 n$.
By Lemma~\ref{lem:r-division}, this takes $O(n)$ time.
We also build a skeleton graph $G_s$, by taking a skeleton set $V_s := V_e \cup \boundary{\rdiv}$.

For each region of $\rdiv$, we run the assumed decremental connectivity algorithm.
Observe that in the proof of Lemma~\ref{lem:first_recursion}, we only need these algorithms to explicitly maintain cc-identifiers of vertices of $V_s$.
Thus, the set of explicit vertices for an algorithm run in a region $R$ is $V_s \cap V(R)$.
The decremental algorithm run for $R$ will maintain local cc-identifiers of these vertices.

We maintain the global cc-identifiers in the skeleton graph $G_s$ in the same way as in the proof of Lemma~\ref{lem:first_recursion}.
The only difference is that now the skeleton set $V_s$ is bigger.
Since $V_s = V_e \cup \boundary{\rdiv}$, this requires $O(n + |V_s| \log n) = O(n + (|V_e| + n / \sqrt{r}) \log n) = O(n + (|V_e| + n / \log n)\log n) = O(n + |V_e| \log n)$ time.
Thus, the update time is $O(n + |V_e| \log n + n \cdot f(\log^2 n) / \log^2 n)$.

Since the cc-identifiers of vertices of $G_s$ are maintained explicitly, in particular we explicitly maintain the cc-identifiers of vertices of $V_e$.
It remains to describe the process of computing the global cc-identifier of an arbitrary vertex $v \in V$.
We first query the decremental algorithm that is run for the region $R$ containing $v$ (in case $v$ is a boundary vertex, we may use an arbitrary region) to obtain the local cc-identifier of $v$.
We check whether there exists a vertex in $V_s \cap V(R)$ that has the same local cc-identifier as $v$.
Since the local cc-identifiers of elements of $V_s \cap V(R)$ are maintained explicitly, at no additional overhead we may simply maintain lists of these vertices, grouped by their local cc-identifier.
If there is a vertex among $V_s \cap V(R)$ with the same local cc-identifier as $v$, we return its global cc-identifier (maintained explicitly).
Otherwise, we return a new cc-identifier by encoding as an integer a pair consisting of the identifier of the region containing $v$ (this requires $O(\log n)$ bits) and the local cc-identifier of $v$ (which requires $O(\log \log n)$ bits).
Thus, obtaining a cc-identifier of an arbitrary vertex requires $g(\log^2 n) + O(1)$ time.
\end{proof}

In order to obtain a faster algorithm we use Lemma~\ref{lem:second_recursion} multiple times.
We prove inductively that for $t = 1, 2, \ldots$ there exists an algorithm $A_t$ which processes updates in $O(t n + n \log^{(t)} n + |V_e| \log n)$ time and returns the cc-identifier of any vertex in $O(t)$ time.
The basis of the induction (algorithm $A_1$) is the algorithm of Lemma~\ref{lem:n_log_n} that maintains cc-identifiers explicitly.
Now, consider $t > 1$, and denote by $f_t(n)$ the running time of algorithm $A_t$. 
We construct algorithm $A_t$ by applying Lemma~\ref{lem:second_recursion} to $A_{t-1}$.
The total update time is
\begin{align*}
O(&n + |V_e| \log n + n\cdot f_{t-1}(\log^2 n)/\log^2 n) \\
  &   = O(n + |V_e| \log n + n/\log^2 n ((t-1) \log^2 n + \log^2 n \log^{(t-1)} \log^2 n))\\
  &   = O(n + |V_e| \log n + n((t-1) + \log^{(t-1)} \log^2 n))\\
        & = O(t n + |V_e|\log n + n \log^{(t-1)} \log^2 n)\\
        & = O(t n + |V_e|\log n + n \log^{(t)} n)
\end{align*}

For $t = \log^{*}n$ and $V_e = \emptyset$ we obtain an algorithm that processes all updates in $O(n \log^{*} n)$ time and answers queries in $O(\log^{*} n)$ time.

From the formal point of view, some comment regarding the recursion is necessary.
When applying Lemma~\ref{lem:second_recursion}, we reduce the problem of maintaining connectivity in a graph on $n$ vertices to a collection of $O(n / r)$ graphs of size at most $r$.
In theory, this statement includes the case when the total size of all graphs increases by a constant factor of $2$, every time we apply the recursion.
However, this cannot happen, as we divide the graph using an $r$-division.
In particular, this means that when creating smaller subproblems we partition the edges of the graph.
In the following section, when we show the main result of the paper, we apply Lemma~\ref{lem:second_recursion} only twice, so this explanation is not necessary.

\section{$O(n)$ time algorithm}
\label{sec:n}
In this section we finally show an algorithm that runs in $O(n)$ time.
We view Lemma~\ref{lem:second_recursion} as a reduction from the problem of maintaining connectivity in a graph of size $n$ to the same problem in a collection of graphs of size $\log^2 n$, whose total size is $O(n)$.
The algorithm ran for a region $R$ is given as the explicit set the set $V_e \cap V(R)$.
Moreover, the query time increases by a constant.
This reduction has an overhead of $O(n + |V_e| \log n)$.

If we use $V_e = \emptyset$, and apply this reduction twice we obtain that in order to maintain connectivity in an $n$-vertex graph, we can maintain connectivity in graphs of at most $O(\log^2 \log n)$ vertices and total size $O(n)$.
We also pay $O(n)$ for this reduction.
However, the number of graphs on at most $O(\log^2 \log n)$ vertices is so small that we can simply precompute their connected components.

\begin{lemma}
\label{lem:tiny_graphs}
Let $w$ be the word size and $\log n \leq w$.
After preprocessing in $o(n)$ time, we may repeatedly initialize and run algorithms for decremental maintenance of connected components in graphs of size $t = O(\log^2 \log n)$.
These algorithms may be given a set of vertices $V_e$, and maintain the cc-identifiers of vertices of $V_e$ explicitly.
An algorithm for a graph of size $t$ runs in $O(t + |V_e|\log t)$ time and may return the cc-identifier of every vertex in $O(1)$ time.
\end{lemma}

\begin{proof}
We will call the set $V_e$ the \emph{explicit set}.
The state of the algorithm is uniquely described by the current set of edges in the graph and the explicit set.
There are $2^{t(t-1)/2}$ labeled graphs on $t$ vertices (including non-planar graphs) and $O(2^t)$ possible explicit sets.
Thus, there are $O(2^{t^2})$ possible states, which, for $t = O(\log^2 \log n)$ gives $2^{O(\log^4 \log n)} = 2^{o(\log n)} = o(n)$.
In particular, each state can be encoded as a binary string of length $O(\log^4 \log n)$ which fits in a single machine word.

For each state, we precompute cc-identifiers.
Moreover, for each pair of state and an edge to be deleted, we compute the changes to the cc-identifiers of vertices in the explicit set.
Observe that if the edge deletion is critical, we simply need to compute the set of vertices in the smaller out of the two connected components that are created and store the intersection of this set and $V_e$.
These vertices should be assigned new, unique cc-identifiers.

We encode the graph by a binary word of length $O(\log^4 \log n)$, where each bit represents an edge between some pair of vertices.
Thus, when an edge is deleted, we may compute the new state of the algorithm in constant time by switching off a single bit.
For any planar graph and any sequence of deletions, the number of changes of cc-identifiers of vertices of $V_e$ is $O(|V_e| \log n)$ (using the analysis similar to the one from the proof of Lemma~\ref{lem:n_log_n}).
The query time is constant, since the cc-identifiers are maintained explicitly.
For each of the $2^{O(\log^4 \log n)}$ states, we require $O(\log^4 \log n)$ preprocessing time.
Thus, the preprocessing time is $o(n)$.
\end{proof}

By applying Lemma~\ref{lem:second_recursion} to the algorithm of Lemma~\ref{lem:tiny_graphs}, and then applying Lemma~\ref{lem:second_recursion} to the resulting algorithm we obtain the main result of the paper.

\begin{theorem}
There exists a decremental connectivity algorithm for planar graphs that supports updates in $O(n)$ total time and answers queries in constant time.
\end{theorem}

\section{Conclusion and open problems}
\label{sec:discussion}
We have shown a reduction from the decremental connectivity problem in planar graphs to incremental connectivity.
As a result, we obtain an algorithm for decremental connectivity that processes all updates in optimal $O(n)$ time and answers queries in
constant time. This shows that the total time complexity of the deceremental problem is not $\Omega(n \log n)$, which seemed to be a natural
bound. In other words we have shown that a lower bound of $\Omega(n \log n)$, that would be an analogous to the lower bound in~\cite{Patrascu06},
cannot hold for decremental algorithms in  planar graphs. We actually conjecture that even for general graphs there exists an $o(n \log n)$ time decremental algorithm.

An interesting question would be to study the worst-case time complexity of decremental connectivity in planar graphs, which has not been fully understood yet. And, contrary to the incremental problem, no nontrivial lower bounds are known.

\bibliographystyle{plain}
\bibliography{references}

\begin{thebibliography}{10}

\bibitem{Alstrup97}
Stephen Alstrup, Jens~P. Secher, and Maz Spork.
\newblock Optimal on-line decremental connectivity in trees.
\newblock {\em Inf. Process. Lett.}, 64(4):161--164, 1997.

\bibitem{Eppstein}
David Eppstein, Zvi Galil, Giuseppe~F. Italiano, and Amnon Nissenzweig.
\newblock Sparsification - a technique for speeding up dynamic graph
  algorithms.
\newblock {\em J. ACM}, 44:669--696, 1997.

\bibitem{Eppstein96}
David Eppstein, Zvi Galil, Giuseppe~F. Italiano, and Thomas~H. Spencer.
\newblock Separator based sparsification: {I}. {P}lanarity testing and minimum
  spanning trees.
\newblock {\em J. Comput. Syst. Sci.}, 52(1):3--27, 1996.

\bibitem{Eppstein92}
David Eppstein, Giuseppe~F. Italiano, Roberto Tamassia, Robert~Endre Tarjan,
  Jeffery Westbrook, and Moti Yung.
\newblock Maintenance of a minimum spanning forest in a dynamic plane graph.
\newblock {\em J. Algorithms}, 13(1):33--54, 1992.

\bibitem{Frederickson85}
Greg~N. Frederickson.
\newblock Data structures for on-line updating of minimum spanning trees, with
  applications.
\newblock {\em {SIAM} J. Comput.}, 14(4):781--798, 1985.

\bibitem{Gustedt}
Jens Gustedt.
\newblock Efficient union-find for planar graphs and other sparse graph
  classes.
\newblock {\em Theoretical Computer Science}, 203(1):123 -- 141, 1998.

\bibitem{Henzinger99}
Monika~R. Henzinger and Valerie King.
\newblock Randomized fully dynamic graph algorithms with polylogarithmic time
  per operation.
\newblock {\em J. ACM}, 46(4):502--516, July 1999.

\bibitem{Henzinger98}
Monika~Rauch Henzinger and Michael~L. Fredman.
\newblock Lower bounds for fully dynamic connectivity problems in graphs.
\newblock {\em Algorithmica}, 22(3):351--362, 1998.

\bibitem{Holm01}
Jacob Holm, Kristian de~Lichtenberg, and Mikkel Thorup.
\newblock Poly-logarithmic deterministic fully-dynamic algorithms for
  connectivity, minimum spanning tree, 2-edge, and biconnectivity.
\newblock {\em J. ACM}, 48(4):723--760, 2001.

\bibitem{Kapron13}
Bruce~M. Kapron, Valerie King, and Ben Mountjoy.
\newblock Dynamic graph connectivity in polylogarithmic worst case time.
\newblock In {\em Proceedings of the Twenty-Fourth Annual ACM-SIAM Symposium on
  Discrete Algorithms}, SODA '13, pages 1131--1142. SIAM, 2013.

\bibitem{Klein-13}
Philip~N. Klein, Shay Mozes, and Christian Sommer.
\newblock Structured recursive separator decompositions for planar graphs in
  linear time.
\newblock In Dan Boneh, Tim Roughgarden, and Joan Feigenbaum, editors, {\em
  Symposium on Theory of Computing Conference, STOC'13, Palo Alto, CA, USA,
  June 1-4, 2013}, pages 505--514. {ACM}, 2013.

\bibitem{Patrascu06}
Mihai P{\v a}tra{\c s}cu and Erik~D. Demaine.
\newblock Logarithmic lower bounds in the cell-probe model.
\newblock {\em {SIAM} J. Comput.}, 35(4):932--963, 2006.

\bibitem{Poutre96}
Johannes A.~La Poutr{\'{e}}.
\newblock Lower bounds for the union-find and the sp;it-find problem on pointer
  machines.
\newblock {\em J. Comput. Syst. Sci.}, 52(1):87--99, 1996.

\bibitem{Tarjan75}
Robert~Endre Tarjan.
\newblock Efficiency of a good but not linear set union algorithm.
\newblock {\em J. {ACM}}, 22(2):215--225, 1975.

\bibitem{Tarjan79}
Robert~Endre Tarjan.
\newblock A class of algorithms which require nonlinear time to maintain
  disjoint sets.
\newblock {\em J. Comput. Syst. Sci.}, 18(2):110--127, 1979.

\bibitem{Thorup99}
Mikkel Thorup.
\newblock Decremental dynamic connectivity.
\newblock {\em J. Algorithms}, 33(2):229--243, 1999.

\bibitem{Thorup00}
Mikkel Thorup.
\newblock Near-optimal fully-dynamic graph connectivity.
\newblock In F.~Frances Yao and Eugene~M. Luks, editors, {\em Proceedings of
  the Thirty-Second Annual {ACM} Symposium on Theory of Computing, May 21-23,
  2000, Portland, OR, {USA}}, pages 343--350. {ACM}, 2000.

\bibitem{Arge-13}
Freek van Walderveen, Norbert Zeh, and Lars Arge.
\newblock Multiway simple cycle separators and {I/O}-efficient algorithms for
  planar graphs.
\newblock In {\em Proceedings of the Twenty-Fourth Annual {ACM-SIAM} Symposium
  on Discrete Algorithms, {SODA} 2013, New Orleans, Louisiana, USA, January
  6-8, 2013}, pages 901--918, 2013.

\bibitem{Nilsen13}
Christian Wulff{-}Nilsen.
\newblock Faster deterministic fully-dynamic graph connectivity.
\newblock In {\em Proceedings of the Twenty-Fourth Annual {ACM-SIAM} Symposium
  on Discrete Algorithms, {SODA} 2013, New Orleans, Louisiana, USA, January
  6-8, 2013}, pages 1757--1769, 2013.

\end{thebibliography}

\end{document}